\documentclass[%
 reprint,
 amsmath,amssymb,
 aps,
]{revtex4-1}
\usepackage{float}
\usepackage{amsthm,amsfonts,amsmath,mathtools}
\newtheorem{theorem}{Theorem}
\makeatletter
\let\oldtheequation\theequation
\renewcommand\tagform@[1]{\maketag@@@{\ignorespaces#1\unskip\@@italiccorr}}
\renewcommand\theequation{(\oldtheequation)}
\makeatother

\usepackage[colorlinks]{hyperref}

\usepackage{graphicx}% Include figure files
\usepackage{dcolumn}% Align table columns on decimal point
\usepackage{bm}% bold math
\usepackage{subcaption}
\newcommand*{\affaddr}[1]{#1} % No op here. Customize it for different styles.
\newcommand*{\affmark}[1][*]{\textsuperscript{#1}}

\usepackage{natbib}
\usepackage{filecontents}

\nocite{*}

\begin{document}

\preprint{APS/123-QED}

\title{Controlled Quantum Search}% Force line breaks with \\
%\thanks{A footnote to the article title}%

\author{%
K. de Lacy\affmark[1] and L. Noakes\affmark[1], J. Twamley\affmark[2], and J.B. Wang\affmark[3]
\\ \vspace{0.2cm}
\affaddr{\affmark[1]\textit{School of Mathematics and Statistics, University of Western Australia, WA 6009, Perth, Australia}
}\\
\affaddr{\affmark[2]\textit{ARC Centre for Engineered Quantum Systems, Department of
Physics and Astronomy, Macquarie University, NSW 2109, Australia}}\\
\affaddr{\affmark[3]\textit{School of Physics and Astrophysics, University of Western Australia, WA 6009, Perth, Australia}}\\
}

\begin{abstract}
Quantum searching for one of \(N\) marked items in an unsorted database of \(n\) items is solved in \(\mathcal{O}(\sqrt{n/N})\) steps using Grover's algorithm. Using nonlinear quantum dynamics with a Gross-Pitaevskii type quadratic nonlinearity, Childs and Young discovered an unstructured quantum search algorithm with a complexity \(\mathcal{O}( \min \{ 1/g \, \log (g n), \sqrt{n} \} ) \), which can be used to find a marked item after \(o(\log(n))\) repetitions, where \(g\) is the nonlinearity strength \cite{PhysRevA.93.022314}. In this work we develop a structured search on a complete graph using a time dependent nonlinearity which obtains one of the \(N\) marked items with certainty.  The protocol has runtime \(\mathcal{O}((N^{\perp} - N) / (G \sqrt{N N^{\perp}}) )\) if \(N^{\perp} > N\), where \(N^{\perp}\) denotes the number of unmarked items and \(G\) is related to the time dependent nonlinearity. If \(N^{\perp} \leq N\), we obtain a runtime \(\mathcal{O}( 1 )\). We also extend the analysis to a quantum search on general symmetric graphs and can greatly simplify the resulting equations when the graph diameter is less than \(5\).
\end{abstract}

\pacs{Valid PACS appear here}% PACS, the Physics and Astronomy
                             % Classification Scheme.
%\keywords{Suggested keywords}%Use showkeys class option if keyword
                              %display desired
\maketitle

%\tableofcontents

\section{\label{sec:level1}INTRODUCTION}

 Using linear quantum mechanics the search problem can be solved using Grover's algorithm \cite{grover1997quantum} in \(\mathcal{O}(\sqrt{n/N})\) steps, where \(n\) denotes the number of search items and \(N\) denotes the number of marked items. Grover's search is asymptotically optimal in the linear quantum domain \cite{bennett1997strengths}.

 The linearity of quantum mechanics plays a subtle but profound role in the design and performance of quantum algorithms. It was shown by Abrams and Lloyd \cite{abrams1998nonlinear} that nonlinear quantum mechanics has the potential to solve NP-complete (nondeterministic polynomial time) and \(\# P\) problems (including oracle problems) in polynomial time.

 Meyer and Wong \cite{meyer2013nonlinear}, and Kahou and Fedor \cite{kahou2013quantum} looked at using the Gross-Pitaevskii dynamics of interacting Bose Einstein condensates to perform Grover's search and found a runtime which scales as \(\mathcal{O}( \min \{ \sqrt{n/g} , \sqrt{n} \} )\), where \(g\) denotes the nonlinearity strength. Meyer and Wong then considered the more general type of nonlinearity \(\sim \, f( \| \phi \|^2 ) \), where \(f : \mathbb{R} \rightarrow \mathbb{R}\) is smooth \cite{meyer2014quantum}. More recently Childs and Young \cite{PhysRevA.93.022314}, found a nonlinear protocol with a runtime scaling as \(\mathcal{O}  (\min \{ 1 / g \log (g n) ; \sqrt{n} \} ) \), which is exponentially faster than previous results \cite{meyer2013nonlinear}. Furthermore this nonlinear search can be repeated \(\text{o}(\log(n))\) times to find the position of a marked item. In all these works however, the marking of the item \( | j^* \rangle \), is performed via the linear part of the dynamics through a term in the Hamiltonian  \(\sim - | j^* \rangle\langle j^* |  \). In our work we consider the case where the marking is encoded into the degree of the nonlinearity and the nonlinearity itself of each item. We consider the case of quantum nonlinear dynamics on a complete graph of \(n\) sites where the initial site is a uniform superposition up to a phase on the marked site, namely \( | \phi (t=0) \rangle = \mathrm{i} \sum_{j = i^* } | j \rangle + \sum_{j \neq i^* } |j \rangle \), ignoring normalisation, where \(i^*\) denotes marked items. We apply a time dependent modulation of the nonlinear strength \(u_k (t)\) to the \(k\)'th state. Although we use a model where both the nonlinear strength and the nonlinearity may depend upon the state, only one is required to depend explicitly upon the states without impacting the end time. This implies our protocol will have the same runtime when governed by linear or nonlinear quantum mechanics. Hence for \(N \ll n\) we obtain the same complexity as Grover, which is asymptotically optimal in the linear case.

 We show analytically that with a suitable form for the nonlinearity strength of the \(k\)'th item, \(u_k (t)\), the protocol yields complete localisation of the quantum dynamics onto the marked states in time \(\mathcal{O}((N^{\perp} - N) / (G \sqrt{N N^{\perp}}) )\), for \(N^{\perp} > N\) and time \( \mathcal{O}( 1 )\) when \(N^{\perp} \leq  N\). The nonlinearity of marked and unmarked items is algebraically related to \(G\) in section \ref{complete graph}. %The speed of our protocol is comparable to the nonlinear search speed of \cite{meyer2013nonlinear, kahou2013quantum}, where \(\sqrt{g'} \sim g\) for constant \(N\). When using an oracle, linear quantum mechanics requires applying this oracle at least \(\mathcal{O}(\sqrt{n})\) times \cite{1046335619971001}. Our algorithm reciprocates this bound for constant \(N\).

We interpret the database search problem as a search on a graph governed by continuous time quantum dynamics, arriving at the Discrete Nonlinear Schr\"odinger Equation (DNLSE). By expressing the coefficients of each state in polar form we can decompose quantum states over the nodes in the graph into equivalence classes depending on the connectivity of the nodes representing unmarked and marked items. For the case of the complete graph this reduction greatly simplifies the description of the dynamics. On this graph we are able to develop a new continuous time algorithm which obtains a marked item with certainty. Furthermore, the error associated with measurement becomes arbitrarily small, unlike the previous work by Meyer and Wong \cite{meyer2013nonlinear} where the peak probability becomes increasingly difficult to obtain.

\section{Discrete Nonlinear Schr\"odinger Equation}
Index the \(N\) marked states by \(i^*\) and let the coefficient of state \(j \in \{0,...,n-1\}\) be \(x_j = r_j \mathrm{e}^{\mathrm{i} \theta_j}\), where \(r_j : [0,t_f] \to [0,1] \), \(\theta_j : [0,t_f] \to ( -\pi , \pi ] \), \(t_f \in \mathbb{R}^+\) and \(\mathrm{i}=\sqrt{-1}\). Let the norm be the natural norm over the complex numbers, \(\|V\|^2 \coloneqq V \, \bar{V}\), where the bar denotes conjugation and \(V \in \mathbb{C}\). The norm squared of the \(i\)'th state's coefficient, \(r_i^2\), is the probability of measuring state \(i\). Therefore performing the search equates to evolving the system to maximise \(r_{i^*}^2\). The dynamics of the coefficients are governed by the discrete nonlinear Schr\"odinger equation (DNLSE)
    \begin{align}\label{1dnlse}
	\mathrm{i} \, \dot x_j = \gamma \mathrm{L}_{jk} x_k + u_j \|x_j\|^{2 \zeta_j} x_j \, ,
	\end{align}
	where \(\gamma = G/(n-2N) \) for some constant \(G\). Using \(x_j = r_j \mathrm{e}^{\mathrm{i} \theta_j}\) and splitting equation \eqref{1dnlse} into its real and imaginary components gives
	\begin{align}  
	\dot r_j & = \gamma \mathrm{L}_{jk} r_k  \sin( \theta_k-\theta_j) \, , \label{DNLSEr} \\
	\dot \theta_j & = - \gamma \mathrm{L}_{jk} \frac{r_k}{r_j}  \cos( \theta_k-\theta_j)  - u_j r_j^{2 \zeta_j} \, , \label{DNLSEt}
	\end{align}
	where the index \(k\) is summed from \(0\) to \(n-1\) in equations \eqref{1dnlse}, \eqref{DNLSEr} and \eqref{DNLSEt}. A dot over a function denotes a derivative with respect to time. The control function and nonlinearity of the \(j\)'th state is \(u_j : [0,t_f] \to \mathbb{R}\) and \(\zeta_j \in \mathbb{Z}\) respectively. We assume that both of these can be manipulated at will and require at least one of \(\zeta_j\) and \(u_j\) to be different for marked and unmarked states. Furthermore we will induce conditions onto \(\zeta_j\) and \(u_j\) with respect to \(j\) so the graph symmetry is preserved in the DNLSE.
    
	   The Laplacian, \(\mathrm{L}\), for an arbitrary graph is formed by taking the graph's adjacency matrix and subtracting the number of connections of the \(j\)'th node from the \(j\)'th element along the diagonal. The number \(\mathrm{L}_{ij}\) denotes the element in the \(i\)'th row and \(j\)'th column of the Laplacian.
       
       Initially all states are prepared with coefficients \(r_j = 1/\sqrt{n}\) for all \(j \in \{0,...,n-1\}\) and \(\theta_i = \theta_j\) or \( \theta_{j^*} + \pi/2 \) for all \(i\), where \(j\) and \(j^*\) indicate unmarked  and marked states respectively. This initial state can be prepared using a controlled rotation on an equal superposition with a linear quantum computer using \(\mathcal{O}(\log(n) ) \) elementary quantum gates.
      
	On a general graph, finding the optimal control curves \(u_j\) to maximise \(r_{i^*}^2\) results in a boundary value problem. We provide a direct numerical method to solve this and for diameter \(3\) and \(4\) graphs, the boundary value problem can be turned into an initial value problem. For complete graphs we obtain analytic expressions for the controls and end time.
		\begin{theorem}\label{thrm:1}
		The DNLSE must conserve the probability of measuring any state, hence
		\begin{align}\label{forpf1}
		\sum_{j=1}^{n} r_j^2 = 1 \, ,
		\end{align}
		when normalised.
	\end{theorem}
	\begin{proof}
		Take the derivative of the left hand side of \eqref{forpf1} with respect to time and substitute equation \eqref{DNLSEr}. To remain physical, the graph must be undirected so \(\mathrm{L}_{jk} = \mathrm{L}_{kj}\). Hence,
		\begin{align}
		\sum_{j=1}^{n} r_j \dot{r}_j  &= \sum_{j,k=1}^{n} \gamma \mathrm{L}_{jk} r_k r_j  \sin( \theta_k-\theta_j) =  0 \, .
		\end{align}
		Integrating this gives equation \eqref{forpf1} under the assumption that the system is normalised.
	\end{proof}
\section*{Reducing the DNLSE via Graph Symmetry}
A reduction based upon symmetry can be performed to simplify the DNLSE. Let all \(n\) nodes of the graph form the set denoted by \(\mathcal{N}\). The set \(\mathcal{N}\) is isomorphic to the set of states labelled by \(\{0,1,..., n-1\}\). The bijective mapping \(\phi: \mathcal{N} \to \{0,1,..., n-1\}\) uniquely identifies each node with a state.
    
    The distance \(d(a,b)\) between two nodes \(a,b \in \mathcal{N} \) is the minimum number of edges in any path connecting \(a\) to \(b\). Two nodes \(a,b\in \mathcal{N}\) are said to be equivalent, \(a \sim b \), if \(\phi(a)\) and \(\phi(b)\) are labels for both marked or both unmarked states, and there exist elements \(c_1,c_2 \in [e]\), where \(d(a,c_1)=d(b,c_2)\) for all \(e \in \mathcal{N}\). Furthermore the set \([e]\), for \(e \in \mathcal{N} \) is defined as \(    [e] \coloneqq \{ c \in \mathcal{N} \, |  \, c \sim e  \} \, , \) called the equivalence class of \(e\). When all nodes in each equivalence class are given the same nonlinearity, then for \(a,b \in [e]\), with \(e\in \mathcal{N} \), the coefficients of states labelled by \(\phi(a) \) and \(\phi(b)\) are equal. Hence, the DNSE can be written using coefficients of one state from each equivalence class under the mapping \(\phi\). Call the process of writing an equation in terms of single elements of equivalence classes a reduction.
    
    \section{Complete Graph} \label{complete graph}
    A complete graph is a graph with every node connected to every other by a unique edge. On a complete graph any state can be directly transformed into any other, hence this is the least restrictive graph possible. To preserve the symmetry of a complete graph, let all marked states have the same nonlinearity, \(\zeta_{*}\), and all unmarked states have the same nonlinearity, \(\zeta\), where \(\zeta_{*} = \zeta\) is allowed.

    For the complete graph, there are only two equivalence classes under our equivalence relation, namely the set of nodes corresponding to marked states and the set of all nodes corresponding to unmarked states. Therefore the reduction process results in a single node representing a marked state, connected to a single node representing an unmarked state.

    If \(N=n\) there is certainty of measuring a marked state. For \(N<n\) marked states, the reduction can be written in terms of the constraints: \(r_{i} = r_{j}, \, \theta_{i} = \theta_{j}\) where \( i\) and \(j\) index marked states, and \(r_{i} = r_{j} , \, \theta_{i} = \theta_{j} \) where \(i\) and \(j\) index unmarked states. The Laplacian for an undirected, complete graph of \(n\) nodes is, 
	\begin{align}
		\mathrm{L} = \left( \begin{array}{ccccc}
		1-n & 1 & 1 & \dots & 1 \\
		1 & 1-n & 1 & \dots & 1 \\
		1 & 1 & 1-n & \dots & 1 \\
		\vdots & \vdots & \vdots & \ddots & \vdots \\
		1 & 1 & 1 & \dots & 1-n
		\end{array} \right) \, .
	\end{align}
	Simplifying the DNLSE in equations \eqref{DNLSEr} and \eqref{DNLSEt} by performing a reduction gives
	\begin{align*}
	\dot r_{*} & = \gamma	(n-N) r \sin( \theta - \theta_{*}) \, ,\\
	\dot r & = \gamma 	N \, r_{*}  \sin(  \theta_{*} - \theta ) \, ,\\
	\dot \theta_{*} & = 	- \gamma ( N-n+ (n-N) \frac{r}{r_{*}}  \cos( \theta-\theta_{*}) ) - u_* r_{*}^{2 \zeta_{*}} \, , \\
	\dot \theta & = - \gamma N \left( \frac{r_{*}}{r}  \cos( \theta_{*}-\theta) +1 \right) - u r^{2 \zeta} \, ,
	\end{align*}
    where \(r\) and \(\theta\) describe the radial and angular components of the coefficient of any unmarked state and \(r_*\) and \(\theta_*\) denote the radial and angular components of the coefficient of any marked state. Similarly all controls for the marked states are denoted \(u\) and all controls for the unmarked states are \(u_*\).

\section*{Controlled Quantum Search on a Complete Graph}
Theorem \eqref{thrm:1} states that the total probability is conserved, which can be rearranged to give
    \begin{align}\label{probcons}
		r = \sqrt{\frac{ 1-Nr_{*}^2}{n-N}}  \, .
	\end{align}
	Therefore the DNLSE can be written without \(r\). Only \(\Theta = \theta - \theta_{*} \) is found in the equation for \(\dot{r}_{*}\), not \(\theta_{*}\) and \(\theta\) separately. Hence the states can be contracted
	\begin{align}\label{10}
		& \dot r_{*}  =  \frac{g}{n-2N}	(n-N) r \sin( \Theta) \, , \\
		& \begin{array}{r}
		\dot \Theta = \displaystyle\frac{g}{n-2N} ((n-N) \, \frac{r}{r_{*}} -N \,  \frac{r_{*}}{r}  )\cos( \Theta )  \\   - g - u r^{2 \zeta}  + u_* r_{*}^{2 \zeta_{*}}  \, .
		\end{array} \label{11}
	\end{align}
    The desired dynamics is for \(r_{*}^2\) to increase as quickly as possible. Therefore the magnitude of \(r \sin(\Theta)\) should be maximised, hence \( \sin(\Theta) = 1 \equiv \Theta = \pi/2 + 2 C_1 \pi\), where \(C_1\) can be set to zero without loss of generality. The initial state, constructed earlier, satisfies this optimality constraint. However, to remain optimal we require \( \Theta =  \pi/2 \) for all time. This turns the differential equation for \(\Theta\) into an algebraic equation that provides a necessary and sufficient condition for the controls to maximise the probability of measuring a marked state in minimal time,
    \begin{align}\label{control}
	u_* r_{*}^{2 \zeta_{*}} - u r^{2 \zeta} = g \, ,
	\end{align}
    where the radial components are known explicitly by equations \eqref{probcons} and \eqref{good}. The differential equation for the radial component is \(\dot{r}_{*}  = (n-N) r\). Integrating this and using the initial condition \(r_{*}(0) = 1/\sqrt{n}\) gives
	\begin{align}\label{good}
		r_{*} = \frac{1}{\sqrt{N}} \, \sin \left( \frac{g \sqrt{N(n-N)} \, t}{n-2N} 
         + \sin^{-1} \left(  \sqrt{\frac{N}{n}}  \right)
        \right)
		\, .
	\end{align}
	The accumulated probability of all marked states is \(Nr_{*}^2\). The shape of this curve is the square of a sine function. In Meyer and Wong's work \cite{meyer2013nonlinear} on solving structured search problems via nonlinear quantum mechanics, they obtain peaks which become arbitrarily narrow, and therefore arbitrarily difficult to measure. In our scheme, the ability to measure a marked item with certainty becomes easier as \(n\) increases because the neighbourhood about \(Nr_{*}^2 = 1\) becomes flatter. Hence the error associated with measurement is essentially negligible for large \(n\). Two plots of the accumulated probability in figure \ref{fig:probability} depict the probability of measuring a marked state as a function of time, for \(n=3\) and \(n=10\), with one marked state and \(g=1\).

	\begin{figure}[H]
		\centering
		\begin{minipage}{0.4\textwidth}
			\centering
			\includegraphics[width=0.9\linewidth]{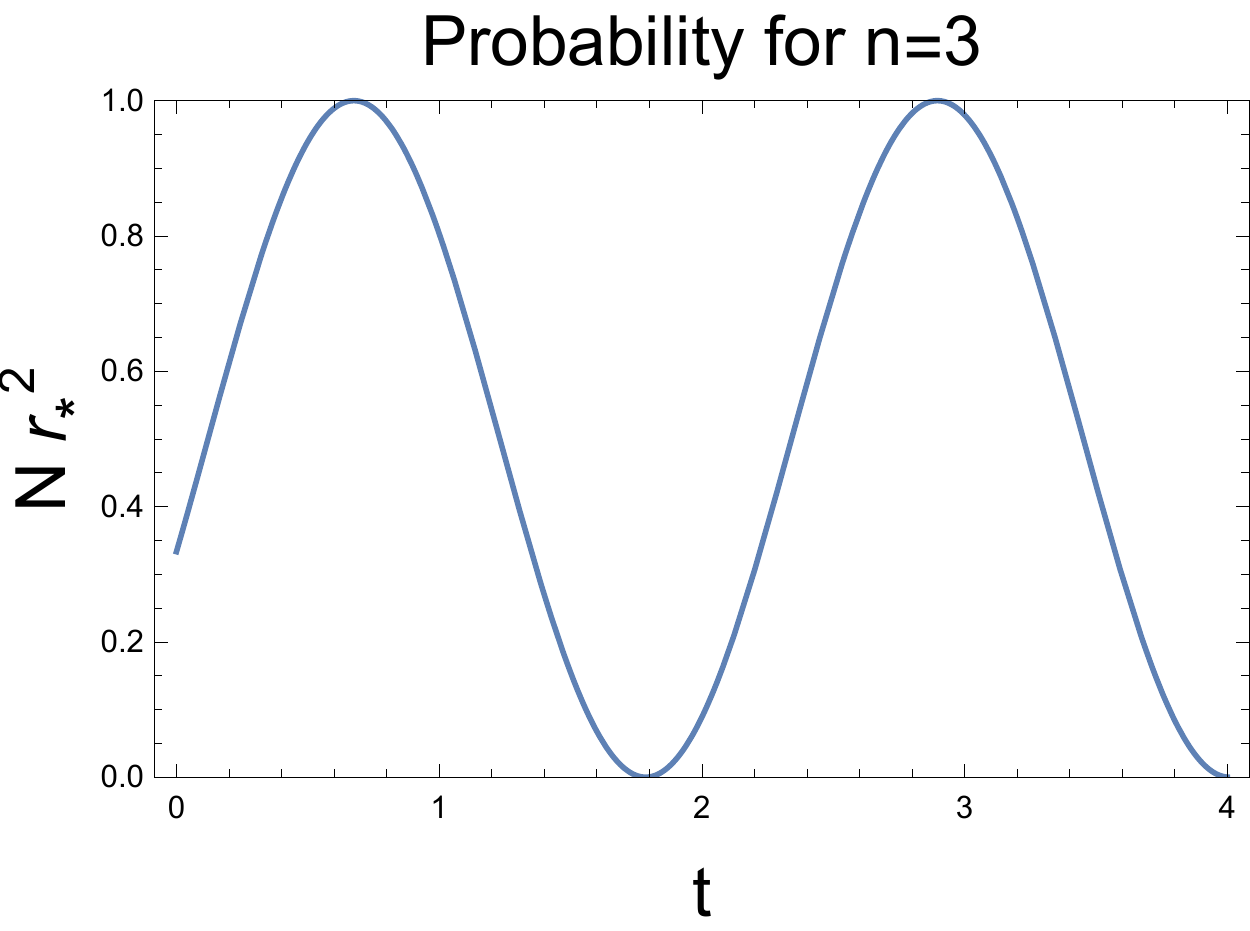}
			\subcaption{}
			\label{fig:Pn3}
		\end{minipage}
		\begin{minipage}{0.4\textwidth}
			\centering
			\includegraphics[width=0.9\linewidth]{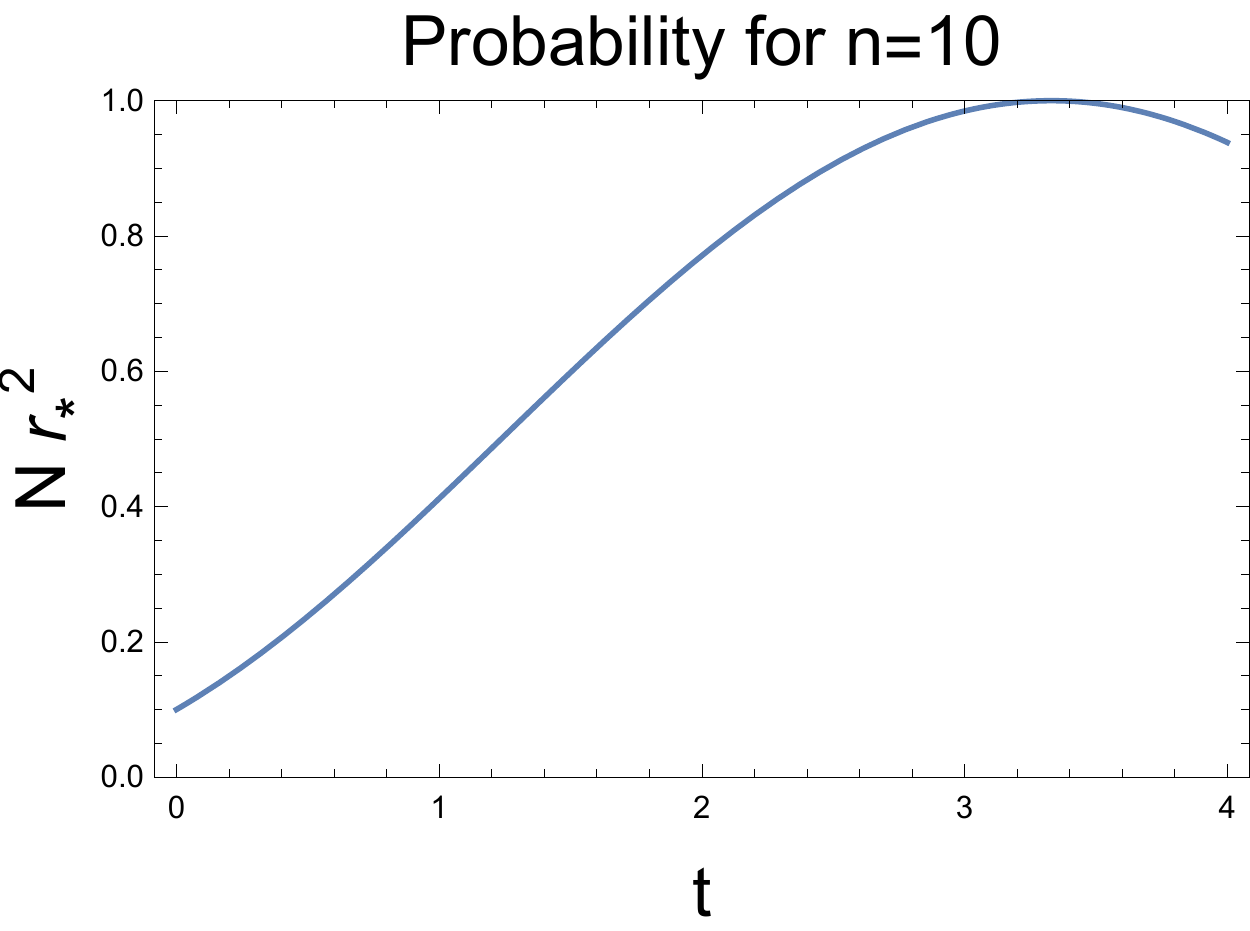}
			\subcaption{}
			\label{fig:Pn10}
		\end{minipage}
		\caption{Each subfigure depicts the probability of measuring a marked state \(Nr_{*}^2\) with respect to time, where \(r_{*}^2\) is determined by equation \ref{good}. Both subfigures have \(g=1\), \(N=1\) but varying \(n\). This variation changes the end time and causes the curve to become flatter around the maximum, hence measurement of the maximum incurs less error as \(n\) increases.} \label{fig:probability}
	\end{figure}

	The terminal condition reads \(0 = \dot{r}_{*}(t_f) r_{*}(t_f)\). As we seek a maximum this condition becomes \(0 = \dot{r}_{*}(t_f)\), solving this for \(t_f\) provides
	\begin{align}\label{end}
	t_f = \frac{n-2N}{g} \frac{ \cos^{-1} \Big( \sqrt{\frac{N}{n}} \Big)}{\sqrt{N(n-N)}} = \mathcal{O}\bigg(	\frac{n-2N}{g \sqrt{N(n-N)}}	\bigg) \, ,
	\end{align}
	using the big-O convention \cite{knuth1976big}. Note that the maximality condition \(0 = \dot{r}_{*}(t_f)\) is equivalent to \(r(t_f) = 0\), which implies that there is zero probability of measuring an unmarked node at time \(t_f\). 
    
%    To find the complexity of this algorithm, consider partitioning each unit of time into \(h\) equal segments. The control \(u\) is synthesized by a constant number of elementary quantum gates during each step. In this context an elementary quantum gate is a transformation which can be implemented via a single physical transformation. The control requires \(\mathcal{O}(t_f) = \mathcal{O}(1/\sqrt{N(n-N)} \, )\) time steps. As each step has constant complexity, the algorithm has \(\mathcal{O}(1/\sqrt{N(n-N)} \, )\) complexity. 
    
     %To compare this search protocol with others, the time is rescaled to make the controls independent of \(n\) and \(N\). It is convenient to define \(N^{\perp} \coloneqq n - N\) as the number of unmarked nodes. We rescale time by \( t \rightarrow t g /(N^{\perp} - N)\), where \(N^\perp \neq N\), and \(g \in \mathrm{R}^+\) is any chosen function of \(n\) and \(N\). Rescaling the Laplacian and controls gives \( L \rightarrow L (N^{\perp} - N) / g\) and \(u_i \rightarrow u_i (N^{\perp} - N) / g\) respectively. This implies \(u_* r_{*}^{2 \zeta_{*}} - u r^{2 \zeta} = g \) while \(N^{\perp} \neq N\). By choosing \(u=0\) and \(\zeta_* = 0\) then \(u_* = g\). The end time becomes \(\mathcal{O}((N^{\perp} - N) / (g \sqrt{N N^{\perp}}) )\) for \(N^{\perp} > N\). When \(N^{\perp} \gg N\), this runtime becomes \(\mathcal{O}( \sqrt{n/N} )\).
     
     When \(N^{\perp} \leq N\) additional unmarked nodes can be implemented so the number of unmarked and marked nodes is equal. However, this assumes we know the exact number of marked nodes. In this case it is optimal to set the controls to zero, returning to linear quantum mechanics. The complexity in this case is the same as Grover's search and the expected time classically, namely \(\mathcal{O}(1)\) \cite{Success}.
     
     %These methods can be extended to optimisation in logarithmic time. 

     % additional unmarked nodes can be artificially implemented to equate the number of marked and unmarked states, providing a runtime \(\mathcal{O}(1 / 2N)\) which is independent of \(g\) because the nonlinearities can be set to zero. When \(N\) is constant and less than \(N^\perp\) the rescaled end time is \(\mathcal{O}(n  / g \sqrt{n}) = \mathcal{O}(\sqrt{n}/g)\) for some \(g\). For \(N << n\), the runtime is approximately \(\mathcal{O}( n / g \sqrt{n N} ) = \mathcal{O}( \sqrt{n / g^2 N} ) \), which is optimal in the linear case according to \cite{boyer1996tight}. However, as \(N\) and \(N^\perp\) become closer the runtime scales better with respect to \(n\). 
     
%     Defining the complexity to be the algorithm runtime, the complexity is \(\mathcal{O}((N^{\perp} - N) / (g \sqrt{N N^{\perp}}) )\) for \(N^{\perp} > N\) and \(\mathcal{O}(1 / 2N)\) for \(N^{\perp} \leq N\). It may be useful, for small \(N\) to use the method by Childs and Young \cite{PhysRevA.93.022314}, hence these algorithms will at worst take time
%     \begin{align*}
%     \begin{cases}
%     \mathcal{O} \Bigg(	\min \Big\{		\frac{\log(g \, n)}{g}  \log( n) \, , \, \frac{N^{\perp} - N}{g \sqrt{N N^{\perp}}} \Big\} \Bigg) & \text{for   } N^{\perp} > N \\
%     \mathcal{O} \Bigg(	\min \Big\{		\frac{\log(g \, n)}{g}  \log( n) \, , \, \frac{1}{2N} \Big\} \Bigg) & \text{for   } N^{\perp} \leq N
%     \end{cases}
%     \end{align*}
%     to find a marked node.

    On a complete graph we have proven the nonlinearities \(\zeta \) and \(\zeta_*\) affect the control and not the optimal convergence rate. Hence these can be chosen to simplify the control. Note that the nonlinearity is not an integral part of the protocol on a complete graph, hence if \(\zeta = \zeta_* = 0\) we obtain a linear search algorithm with the same convergence rate. 
    
    Define the error \( E \coloneqq 1-Nr_*^2(t_f) \) as the probability of measuring an unmarked state at time \(t_f\) given by equation \eqref{end}. We assume this error only results from the inability to reconstruct the control perfectly in a physical system. Given \(N=1\) and \(\zeta_* = \zeta = 0\), then we could choose controls \(u=0\) and \(u_* = g\). Then assume the control functions are simulated to error \(\nu\) and \(\nu_*\) such that, \(u= \nu \) and \(u_* = g + \nu_* \) for constant \(\nu_*, \nu \in \mathbb{R}\). Then the error decreases as the number of states increases as per figure \eqref{fig:er}.
	\begin{figure}[H]
    \centering
	\includegraphics[width=.45\textwidth]{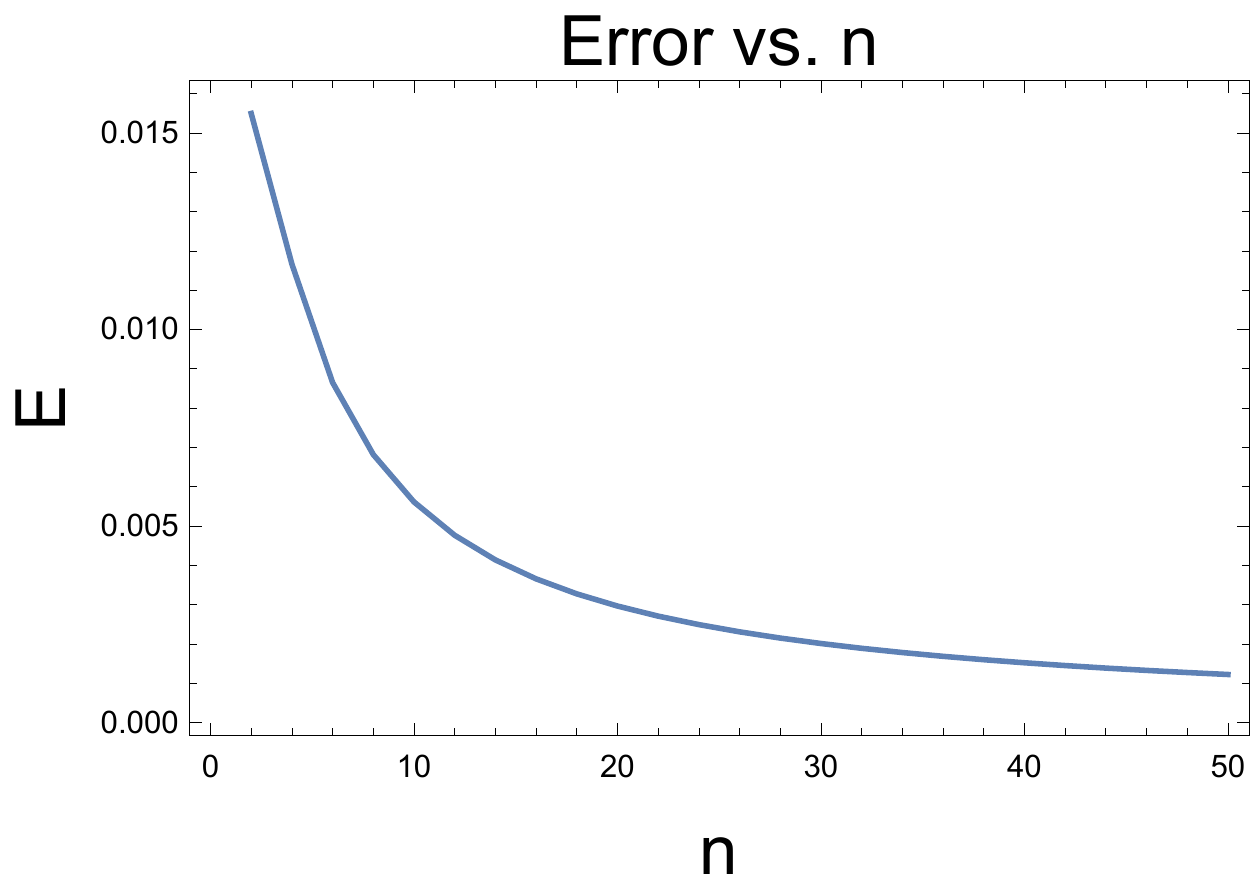}
    \caption{The error at time \(t_f\) as a function of the total number of states. It is assumed the control is incorrectly simulated such that \(\nu_* = \nu = 0.5\). The simulation has one marked state and \(n-1\) unmarked states.}
    \label{fig:er}
    \end{figure}

	\section*{Symmetric Graphs}
    Let there be one marked state, \(N=1\), and consider any symmetric graph \(S\). More precisely, \(S\) is edge and vertex transitive. Let \(d \in \mathbb{N}\) be an integer denoting the diameter of the graph. %Figure \eqref{fig:graphs} shows symmetric graphs with the same number of nodes and connections, but connected differently.

% 	\begin{figure}[H]
% 		\begin{minipage}{0.5\textwidth}
% 			\centering
% 			\includegraphics[width=0.7\linewidth]{images/s2graph1}
% 			\subcaption{}
% 			\label{fig:s2graph1}
% 		\end{minipage}
% 		\begin{minipage}{0.5\textwidth}
% 			\centering
% 			\includegraphics[width=0.7\linewidth]{images/s2graph2}
% 			\subcaption{}
% 			\label{fig:s2graph2}
% 		\end{minipage}
% 	\caption{Different symmetric graphs with the same number of nodes and connectivity.}
% 	\label{fig:graphs}
% 	\end{figure}
	Call the set of nodes with distance \(i\) to the node representing the marked state the \(i\)'th shell. Give every node in the same shell the same nonlinearity. The graph \(S\) can be fully described by: 
    \begin{enumerate}
    \item Its diameter \(d\).
    \item The number of edges from a node on one shell to the next. The number of edges for a node on shell \(i\) to shell \(i+1\) is denoted \(c_i\), where \(i=0,1,2, ... , d-1\).
    \item The number of nodes on each shell. The \(i\)'th shell has \(n_i\) nodes, where \(i=0,1,2, ... , d\). The index \(0\) denotes the node representing the marked state, hence \(n_0=1\).
    \end{enumerate}
    There are particular relations between these parameters and they cannot be chosen arbitrarily. Furthermore the number of connections from a node in disk \(i+1\) to one on disk \(i\) is \(c_i n_i/ n_{i+1} \). The value \(c_0\) denotes the number of edges all other nodes must have to ensure the symmetry is preserved. Therefore the number of edges from a node on the \(i\)'th shell to other nodes on the \(i\)'th shell is \(c_0 - c_{i-1} n_{i-1}/n_{i} - c_{i} \). Upon performing a reduction, each shell forms an equivalence class. Hence the reduction results in one node from each shell. Let \(i\) denote the index of a node in the \(i\)'th shell. Then the DNLSE reads
	\begin{align*} 
	\dot r_0 & = \gamma c_0 r_1  \sin( \theta_1-\theta_0)   \\
	\dot r_j & = \gamma  \left( \frac{c_{j-1} n_{j-1}}{n_{j}} r_{j-1}  \sin( \theta_{j-1}-\theta_j) 
	+c_j r_{j+1}  \sin( \theta_{j+1}-\theta_j) \right) \\
	\dot r_d & = \gamma \frac{c_{d-1} n_{d-1}}{n_{d}} r_{d-1}  \sin( \theta_{d-1}-\theta_d)  \\
	\dot \theta_0 & =\gamma \left( c_0 - c_0 \frac{r_1}{r_0}  \cos( \theta_1-\theta_0) \right)  - u_0 r_0^{2\zeta_0}  \\ 
	\dot \theta_j & = \gamma \left(  -\frac{c_{j-1} n_{j-1}}{n_{j}} \frac{r_{j-1}}{r_j}  \cos( \theta_{j-1}-\theta_j) \right. \\ & \quad  \left.
	 - \big( - \frac{c_{j-1} n_{j-1}}{n_{j}} - c_{j}\big)  
	- c_j \frac{r_{j+1}}{r_j}  \cos( \theta_{j+1}-\theta_j)  \right) - u_j r_j^{2\zeta_j}  \\ 
	\dot \theta_d & = \gamma \left(  -\frac{c_{d-1} n_{d-1}}{n_{d}} \frac{r_{d-1}}{r_d}  \cos( \theta_{d-1}-\theta_d) \right.  \\ & \quad \left.
	- \big( - \frac{c_{d-1} n_{d-1}}{n_{d}} - c_{d}\big) \right)  - u_d r_d^{2\zeta_d}  \, , %\label{big6}
	\end{align*}
	for \(j=1,2,..., d-1\). These equations are rather nasty, however, there are no summations and the number of differential equations has been reduced from \(2n\) to \(2(d+1)\). To maximise the probability of measuring the \(0\)'th state, choose the control to maximise the PMP Hamiltonian 
	\begin{align*}
	\mathcal{H} &= \lambda_j \Big( \gamma	\frac{c_{j-1} n_{j-1}}{n_{j}} r_{j-1}  \sin( \theta_{j-1}-\theta_j) \Big) \\
	& \quad + \gamma \Lambda_j  \Big(
	-\frac{c_{j-1} n_{j-1}}{n_{j}} \frac{r_{j-1}}{r_j}  \cos( \theta_{j-1}-\theta_j) \\ & \qquad \qquad
	- \Big( - \frac{c_{j-1} n_{j-1}}{n_{j}} - c_{j}\Big)  \\ & \qquad \qquad 
	- c_j \frac{r_{j+1}}{r_j}  \cos( \theta_{j+1}-\theta_j) - \frac{u_j}{\gamma} r_j^{2\zeta_j} 
	\Big)  \, ,
	\end{align*} 
	where \(r_{d+1} = 0\), \(c_{-1} = 0\) and \(j\) is summed from \(0\) to \(n-1\). The costates are defined by
	\begin{align*}
		- \frac{\dot{\lambda}_x}{\gamma}  = \frac{1}{\gamma} \frac{\partial \mathcal{H}}{\partial r_x} &= 
		\lambda_{x+1} \Big(	\frac{c_{x} n_{x}}{n_{x+1}}   \sin( \theta_{x}-\theta_{x+1}) \Big) \\ & \quad
		 + \Lambda_{x+1} \Big(
		-\frac{c_{x} n_{x}}{n_{x+1}} \frac{1}{r_{x+1}}  \cos( \theta_{x}-\theta_{x+1}) \Big) \\
		& \quad + \Lambda_x \Big(
		\frac{c_{x-1} n_{x-1}}{n_{x}} \frac{r_{x-1}}{r_x^2}  \cos( \theta_{x-1}-\theta_x) \Big)  \\ & \quad  +
		\Lambda_{x-1} \Big(  
		- c_{x-1} \frac{1}{r_{x-1}}  \cos( \theta_{x}-\theta_{x-1}) \big) \\
		& \quad
		+
		\Lambda_x \Big(  
		 c_x \frac{r_{x+1}}{r_x^2}  \cos( \theta_{x+1}-\theta_x) \big)  \\ & \quad +
		\Lambda_x \Big(
		-2 \zeta_x  \frac{u_x}{\gamma} r_x^{2 \zeta_x-1} \Big) \, ,
	\end{align*}
	and 
	\begin{align*}
		- \frac{\dot{\Lambda}_x}{\gamma}  = \frac{\partial \mathcal{H}}{\partial \theta_x} &=
		- \lambda_x \Big(	\frac{c_{x-1} n_{x-1}}{n_{x}} r_{x-1}  \cos( \theta_{x-1}-\theta_x) \Big)  \\ & \quad
		+ \lambda_{x+1} \Big(	\frac{c_{x} n_{x}}{n_{x+1}} r_{x}  \cos( \theta_{x}-\theta_{x+1}) \Big)
		\\ & \quad 
		+\Lambda_x \Big(
		-\frac{c_{x-1} n_{x-1}}{n_{x}} \frac{r_{x-1}}{r_x} \sin( \theta_{x-1}-\theta_x)  \Big)  \\ & \quad +
		\Lambda_{x+1} \Big(
		\frac{c_{x} n_{x}}{n_{x+1}} \frac{r_{x}}{r_{x+1}} \sin( \theta_{x}-\theta_{x+1}) \Big) 
		\\ & \quad +
		\Lambda_x 
		 \Big(  
		- c_x \frac{r_{x+1}}{r_x}  \sin( \theta_{x+1}-\theta_x) 
		\Big)  \\ & \quad +
		\Lambda_{x-1} 
		\Big(  
		 c_{x-1} \frac{r_{x}}{r_{x-1}}  \sin( \theta_{x}-\theta_{x-1}) 
		\Big) \, .
	\end{align*}
	The optimality condition is 
	\begin{align}
		\Lambda_i r_i^{2\zeta_i} = 0 \, ,
	\end{align}
	where \(i\) is summed from \(0\) to \(d\). This provides a single piece of information. 
\begin{theorem}\label{sum}
		The sum over costates of \(\theta \) is zero,
		\begin{align}\label{forpf2}
		\sum_{i=1}^{n} \Lambda_i = 0 \, .
		\end{align}
	\end{theorem}
	\begin{proof}
		The derivative of the left hand side of equation \eqref{forpf2} is
		\begin{align*}
		- \sum_{i=1}^{n} \dot{\Lambda}_i & = \big( \lambda_j  r_i  - 
		\lambda_i   r_j   \big) \mathrm{L}_{ji} \cos( \theta_j-\theta_i)  \\ & \quad
		+
		\Big( \Lambda_j  \frac{r_i}{r_j} 
		+
		\Lambda_i  \frac{r_j}{r_i}   \Big) \mathrm{L}_{ji}  \sin( \theta_i-\theta_j) \\
% 		& = -\big(	\lambda_i   r_j -  \lambda_j  r_i  \big)
% 		\mathrm{L}_{ij} \cos( \theta_i-\theta_j)  \\ & \quad 
% 		-
% 		\Big( \Lambda_i  \frac{r_j}{r_i} +\Lambda_j  \frac{r_i}{r_j}  \Big)
% 		\mathrm{L}_{ij}  \sin( \theta_j-\theta_i) \\
		&=0 \, .
		\end{align*}
		Integrating this and substituting the transversality conditions for \(\Lambda_i\) gives equation \eqref{forpf2}.
	\end{proof}
    With the equation from Theorem \eqref{sum} and its derivative, along with the extrema condition, three costates can be found as functions of the other costates and states as long as the conditions are independent. Furthermore, only the difference in phase between adjacent shells are important, this can be used to eliminate one state. Furthermore these new conditions can be differentiated to find an additional four conditions on the costates. If these conditions are independent the costates can be determined in terms of the states and control when \(d=2\) or \(3\). In these cases, the control can be written in terms of the states and costates, hence the boundary value differential equations  becomes initial value differential equations which can be solved using a feedback loop. This can be done using a classical computer and there is a significant amount of research aimed at developing techniques to solve forward differential equations using feedback loops in quantum computation \cite{PhysRevLett.85.3045,PhysRevA.62.022108,grimsmo2015time,wang2015quantum}. 
    
	When \(d \geq 4\), we obtain a boundary value differential equation. This can be solved numerically. When the radial component of an unmarked state becomes zero, the phase loses all meaning and the derivative of the phase can easily grow to infinity. To avoid this, Cartesian coordinates are used to find a numerical solution. Furthermore a small amount of error when forward solving the DNLSE will grow extremely rapidly. To reduce this effect we use an adaptive step-size, Runge-Kutta (Radau IIA) method. The nonlinearity can be optimised using a discrete optimiser. The control is constructed from a cubic B-spline.
		\begin{figure}[H]
		\centering
			\includegraphics[width=.6\linewidth]{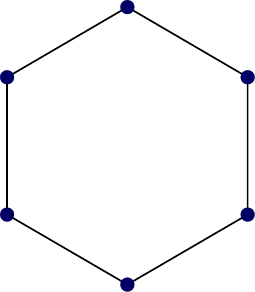}
		\caption{An illustration of the circular graph with \(6\) nodes.}\label{fig:ppic1}
	\end{figure}
	Consider the circular graph of six nodes in figure \eqref{fig:ppic1}. Performing a reduction, this becomes the four node system defined by
	\begin{align*}
	\dot r_0 &= 2 \gamma r_1  \sin( \theta_1-\theta_0)\\
	\dot r_1 &= \gamma r_{0}  \sin( \theta_{0}-\theta_1) 
	+ \gamma r_{2}  \sin( \theta_{2}-\theta_1)\\
	\dot r_2 &= \gamma r_{1}  \sin( \theta_{1}-\theta_2) 
	+ \gamma r_{3}  \sin( \theta_{3}-\theta_2)\\
	\dot r_3 &= 2 \gamma r_{2}  \sin( \theta_{2}-\theta_3)\\
	\dot \theta_0 &= 2 \gamma - 2 \gamma \frac{r_1}{r_0}  \cos( \theta_1-\theta_0) - u r_0^{\zeta_0}\\
	\dot \theta_1 &= - \gamma \frac{r_{0}}{r_1}  \cos( \theta_{0}-\theta_1)  
	+2 \gamma
	- \gamma \frac{r_{2}}{r_1}  \cos( \theta_{2}-\theta_1) - u_1 r_1^{\zeta_1}\\ 
	\dot \theta_2 &= - \gamma \frac{r_{1}}{r_2}  \cos( \theta_{1}-\theta_2)  
	+2 \gamma
	- \gamma \frac{r_{3}}{r_2}  \cos( \theta_{3}-\theta_2) - u_2 r_2^{\zeta_2}\\ 
	\dot \theta_3 &= -2 \gamma \frac{r_{2}}{r_3}  \cos( \theta_{2}-\theta_3)   
	+3 \gamma  - u_3 r_3^{\zeta_3}  \, .
	\end{align*}
	 For convenience we use a nonlinearity \(\zeta_*=1\) on the unmarked states and \(\zeta=2\) on the marked state. The control is described by a finite number of elements by using a cubic B-spline with \(5\) control points. The control points are forced to have magnitude less than \(20\) to ensure the magnitude of the control is always less than \(20\). In practice this bound would be replaced with the physical limitations of the apparatus.
     
	Only the phase differences are important so set \(\theta_1(0)=0\), the remaining initial phases are parameters to be chosen by the numerical optimisation. The solution with the highest probability takes a total time of \(7.70\) seconds and converges with a probability of \(0.98\) to measure the marked state.
	
	After \(1.43\) seconds the first peak of \(r_0^2\), has a height of \(0.95\). This solution is far more practical because it converges almost eight times quicker than the previous solution.

\section*{Summary}
When the entanglement of a quantum system is represented by the DNLSE with a complete graph, we  determine an explicit algorithm to determine the optimal time dependent nonlinearity. The resulting search protocol has runtime \(\mathcal{O}((N^{\perp} - N) / (g \sqrt{N N^{\perp}}) )\) for \(N^{\perp} > N\) and for \(N^{\perp} \leq N\), the runtime is \(\mathcal{O}( 1 )\). This protocol scales equally with Grover's search and can be implemented on a linear or nonlinear quantum computer. Furthermore as the number of states increase the error resulting from measurement decreases.

For a symmetric graph with diameter two or three the resulting boundary value problem can be reduced to an initial value problem. However, for larger diameters, maximising the probability of marked states becomes more complex as it is no longer optimal to set the phase difference between nodes to \(\pi/2\). We develop a direct numerical package to maximise the probability of the marked states subject to the discrete nonlinear Schr\"odinger equation and initial conditions.
	\bibliography{general}
	%\bibliographystyle{plain}
	%\printbibliography

\end{document}